\documentclass{article}
\usepackage[latin9]{inputenc}
\usepackage{mathtools}
\usepackage{amsmath}
\usepackage{amsthm}
\usepackage{amssymb}
\usepackage{graphicx}

\makeatletter
  \theoremstyle{remark}
  \newtheorem{remark}{\protect\remarkname}
  \theoremstyle{plain}
  \newtheorem{lemma}{\protect\lemmaname}
  \theoremstyle{plain}
  \newtheorem{proposition}{\protect\propositionname}
\theoremstyle{plain}
\newtheorem{theorem}{\protect\theoremname}

\usepackage{enumerate}
\usepackage{comment}

  \providecommand{\lemmaname}{Lemma}
  \providecommand{\propositionname}{Proposition}
  \providecommand{\remarkname}{Remark}
\providecommand{\theoremname}{Theorem}

  \providecommand{\lemmaname}{Lemma}
  \providecommand{\propositionname}{Proposition}
  \providecommand{\remarkname}{Remark}
\providecommand{\theoremname}{Theorem}

\makeatother

  \providecommand{\lemmaname}{Lemma}
  \providecommand{\propositionname}{Proposition}
  \providecommand{\remarkname}{Remark}
\providecommand{\theoremname}{Theorem}

\begin{document}

\title{Explicit solutions of the kinetic and potential matching conditions of the energy shaping method}

\author{Sergio D. Grillo\\
 Instituto Balseiro, UNCuyo-CNEA\\
 San Carlos de Bariloche, Río Negro, República Argentina \and Leandro
M. Salomone\\
 CMaLP, Facultad de Ciencias Exactas, UNLP\\
 La Plata, Buenos Aires, República Argentina \and Marcela Zuccalli\\
 CMaLP, Facultad de Ciencias Exactas, UNLP\\
 La Plata, Buenos Aires, República Argentina}
\maketitle
\begin{abstract}
In this paper we present a procedure to integrate, up to quadratures,
the matching conditions of the energy shaping method. We do that in
the context of underactuated Hamiltonian systems defined by simple
Hamiltonian functions. For such systems, the matching conditions split
into two decoupled subsets of equations: the \textit{kinetic} and
\textit{potential} equations. First, assuming that a solution of the
kinetic equation is given, we find integrability and positivity conditions
for the potential equation (because positive-definite solutions are
the interesting ones), and we find an explicit solution of the latter.
Then, in the case of systems with one degree of underactuation, we
find in addition a concrete formula for the general solution of the
kinetic equation. An example is included to illustrate our results. 
\end{abstract}

\section{Introduction}

The \emph{energy shaping} method is a technique for achieving (asymptotic)
stabilization of underactuated Lagrangian and Hamiltonian systems.
See Ref. \cite{CBLMW02} for a review on the subject and \cite{chang4,romero}
for more recent works. In this paper, we shall concentrate on Hamiltonian
systems only. We shall represent underactuated Hamiltonian systems
by pairs $\left(H,\mathcal{W}\right)$, where $H:T^{*}Q\rightarrow\mathbb{R}$
is a Hamiltonian function on a finite-dimensional smooth manifold
$Q$, and $\mathcal{W}$ is a (proper) subbundle of the vertical bundle
of $T^{*}Q$, containing the actuation directions.

The energy shaping method is based on the idea of \emph{feedback equivalence}
\cite{CBLMW02}, and its purpose is to construct, for a given pair
$\left(H,\mathcal{W}\right)$, a state feedback controller and a Lyapunov
function $\hat{H}:T^{*}Q\rightarrow\mathbb{R}$ for the resulting
closed-loop system. To do that, a set of partial differential equations
(PDEs), known as \emph{matching conditions}, must be solved. Such
PDEs have the pair $\left(H,\mathcal{W}\right)$ as datum and the
aforementioned Lyapunov function $\hat{H}$ as their unknown.

In the Chang version of the method \cite{chang1,chang2,chang3,chang4},
on which we shall focus, only simple functions $H$ and $\hat{H}$
are considered (i.e. functions with the \textit{kinetic plus potential
energy} form), and the subbundle $\mathcal{W}$ is assumed to be the
vertical lift of a subbundle $W\subseteq T^{*}Q$. In such a case,
the matching conditions decompose into two subsets: the \textit{kinetic
}and\textit{ potential matching conditions}, or simply, the \textit{kinetic}
and \textit{potential }equations. In canonical coordinates $\left(\mathbf{q},\mathbf{p}\right)$,
if we write (sum over repeated indices is assumed) 
\begin{equation}
H(\mathbf{q},\mathbf{p})=\frac{1}{2}\,p_{i}\,\mathbb{H}^{ij}(\mathbf{q})\,p_{j}+h(\mathbf{q})\label{H}
\end{equation}
and 
\begin{equation}
\hat{H}(\mathbf{q},\mathbf{p})=\frac{1}{2}\,p_{i}\,\hat{\mathbb{H}}^{ij}(\mathbf{q})\,p_{j}+\hat{h}(\mathbf{q}),\label{V}
\end{equation}
such equations read {[}see Ref. \cite{lyap-shap}, Eqs. $\left(56\right)$
and $\left(57\right)${]} 
\begin{equation}
\left(\frac{\partial\hat{\mathbb{H}}^{ij}(\mathbf{q})}{\partial q^{k}}\,\mathbb{H}^{kl}(\mathbf{q})-\frac{\partial\mathbb{H}^{ij}(\mathbf{q})}{\partial q^{k}}\,\hat{\mathbb{H}}^{kl}(\mathbf{q})\right)\,p_{i}p_{j}p_{l}=0\label{kec}
\end{equation}
(the kinetic equation) and 
\begin{equation}
\left(\frac{\partial\hat{h}(\mathbf{q})}{\partial q^{k}}\,\mathbb{H}^{kl}(\mathbf{q})-\frac{\partial h(\mathbf{q})}{\partial q^{k}}\,\hat{\mathbb{H}}^{kl}(\mathbf{q})\right)\,p_{l}=0\label{pec}
\end{equation}
(the potential equation), and they must hold for all $\mathbf{p}\in W_{\mathbf{q}}^{\perp}$
(the orthogonal must be calculated with respect to the metric defined
by $\hat{\mathbb{H}}$). Once a solution $(\hat{\mathbb{H}},\hat{h})$
of these equations is found, the method gives a prescription to construct
a state feedback controller. So, we can say that the core of the method
consists of solving above PDEs.

Let us mention that one looks for a solution $\hat{h}$ of the potential
equation which is positive-definite around some critical point of
$h$. This insures that the function $\hat{H}$ is a Lyapunov function
for the resulting closed-loop system and the given critical point.

In Reference \cite{lewis}, necessary and sufficient conditions were
given for the existence of a solution $\hat{h}$ of the potential
equation, once a solution $\hat{\mathbb{H}}$ of the kinetic equation
is given (note that $\hat{\mathbb{H}}$ can be seen as a datum for
the potential equation). This was done within the framework of the
Goldschmidt's integrability theory for linear partial differential
equations \cite{gold}, which only works in the analytic category.
Nevertheless, no conditions have been presented in order to ensure
the existence of positive-definite solutions. Also, no general recipe
to construct an explicit solution $\hat{h}$ has been developed. The
first goal of the present paper is three-fold: 
\begin{itemize}
\item to extend the results of \cite{lewis} to the $C^{\infty}$ category, 
\item to include positivity conditions, 
\item and to present a systematic procedure to integrate the potential equation
up to quadratures. 
\end{itemize}
Regarding the kinetic equation, few general results are known about
the existence of solutions. In the case of underactuated systems with
one degree of underactuation, the problem was completely solved in
References \cite{chang1,chang2,assym,ghar}. However, a general prescription
for finding explicit solutions is still lacking. The second goal of
this paper is, for one degree of underactuaction, to give such a prescription.

The paper is organized as follows. In Section \S \ref{rewr} we re-define
the unknown $(\hat{\mathbb{H}},\hat{h})$ of the matching conditions.
This gives rise to a new set of equations, in terms of which we shall
study a particular subclass of solutions of the kinetic equation.
Given $\hat{\mathbb{H}}$ inside the mentioned subclass, in Section
\S \ref{solvp} we find sufficient conditions for the existence of
solutions $\hat{h}$ of the potential equation. Also, sufficient conditions
to ensure positive-definiteness of $\hat{h}$ are given. All of these
conditions together give rise to a procedure that enable us to find
an explicit expression for local solutions of the potential equation
(up to quadratures). Finally, we devote Section \S \ref{sec:potkincuad}
to extend the procedure to integrate up to quadratures the kinetic
equation also, but in the particular subclass of underactuated Hamiltonian
systems with only one degree of underactuation. To conclude, we apply
such a procedure to the inverted double pendulum.

\bigskip{}

We assume that the reader is familiar with basic concepts of Differential
Geometry \cite{boot,kn,mrgm}, Hamiltonian systems in the context
of Geometric Mechanics \cite{am,ar,mr} and Control Theory in a geometric
language \cite{bloc,bullo}.

\bigskip{}

\textbf{Basic notation and definitions.} Along all of the paper, $Q$
will denote a smooth connected manifold of dimension $n$ and $\left(TQ,\tau,Q\right)$
and $\left(T^{\ast}Q,\pi,Q\right)$ the tangent bundle and its dual
bundle, respectively. As it is customary, we denote by $\left\langle \cdot,\cdot\right\rangle $
the natural pairing between $T_{q}Q$ and $T_{q}^{\ast}Q$ at every
$q\in Q$ and by $\mathfrak{X}\left(Q\right)$ and $\Omega^{1}\left(Q\right)$
the sheaves of sections of $\tau$ and $\pi$, respectively. If $F:Q\rightarrow P$
is a smooth function between differentiable manifolds, we denote by
$F_{*}$ and $F^{*}$ the push-forward map and its transpose, respectively.

Consider a local chart $\left(U,\varphi\right)$ of $Q$, with $\varphi:U\rightarrow\mathbb{R}^{n}$.
Given $q\in U$, we write $\varphi\left(q\right)=\left(q^{1},\dots,q^{n}\right)=\mathbf{q}$.
For the induced local chart $\left(T^{\ast}U,\left(\varphi^{\ast}\right)^{-1}\right)$
on $T^{\ast}Q$, i.e. the canonical coordinates of the cotangent bundle,
we write, for all $\alpha\in T^{\ast}U$, 
\[
\begin{array}{l}
\left(\varphi^{\ast}\right)^{-1}\left(\alpha\right)=\left(q^{1},\dots,q^{n},p_{1},\dots,p_{n}\right)=\left(\mathbf{q},\mathbf{p}\right),\end{array}
\]
or simply 
\[
\left(\varphi_{q}^{\ast}\right)^{-1}\left(\alpha\right)=\mathbf{p}.
\]

By a quadratic form on a subbundle $V\subseteq T^{*}Q$ we shall understand
a function $\mathfrak{H}:V\rightarrow\mathbb{R}$ given by the formula
\begin{equation}
\mathfrak{H}\left(\alpha\right)=\frac{1}{2}\left\langle \alpha,\rho^{\sharp}\left(\alpha\right)\right\rangle ,\;\;\;\forall\alpha\in V,\label{rmd}
\end{equation}
where $\rho$ is a fibered inner product on $V^{*}$, $\rho^{\flat}:V^{*}\rightarrow V$
is given by 
\[
\left\langle \rho^{\flat}\left(u\right),v\right\rangle =\rho\left(u,v\right),
\]
and $\rho^{\sharp}:V\rightarrow V^{*}$ is the inverse of $\rho^{\flat}$.
Given a subbundle $W\subseteq V$, by $W^{\bot}$ we shall denote
the orthogonal complement w.r.t. $\rho$ and by $W^{\sharp}$ the
subbundle $\rho^{\sharp}\left(W\right)\subseteq V^{*}$. It is easy
to see that 
\begin{equation}
W^{\bot}=\rho^{\sharp}\left(W^{\circ}\right)=\left(W^{\circ}\right)^{\sharp}.\label{rort}
\end{equation}
Here, by $W^{\circ}\subseteq V^{*}$ we are denoting the annihilator
of $W$. We shall also say that $W^{\bot}$ is the orthogonal complement
of $W$ with respect to $\mathfrak{H}$.

The vertical subbundle associated to any vector bundle is given by
the kernel of the push-forward of the bundle projection, and every
element (resp. smooth section) of this subbundle is called vertical
(resp. vertical vector field). In particular, for the cotangent bundle,
it is given by $\ker\pi_{*}\subset TT^{*}Q$. If $\alpha\in T_{q}^{*}Q$,
there is a canonical way to identify $\ker\pi_{*,\alpha}$ and $T_{q}^{*}Q$.
This can be done through the vertical lift map $\mathsf{vlift}_{\alpha}^{\pi}:T_{q}^{\ast}Q\rightarrow\ker\pi_{\ast,\alpha}$,
defined by 
\begin{equation}
\mathsf{vlift}_{\alpha}^{\pi}(\beta)=\left.\frac{\mathrm{d}}{\mathrm{d}t}(\alpha+t\beta)\right\vert _{t=0}.\label{veri}
\end{equation}
This map is in fact an isomorphism of vector spaces for every $\alpha\in T^{*}Q$.

Given a smooth function $f:T^{\ast}Q\rightarrow\mathbb{R}$, we define
the fiber derivative of $f$ as the vector bundle morphism $\mathbb{F}f:T^{\ast}Q\rightarrow TQ$
such that, for every $\alpha,\beta\in T_{q}^{\ast}Q$, 
\begin{equation}
\left\langle \beta,\mathbb{F}f(\alpha)\right\rangle =\left.\frac{\mathrm{d}}{\mathrm{d}t}f(\alpha+t\beta)\right\vert _{t=0}=\left\langle \mathrm{d}f\left(\alpha\right),\mathsf{vlift}_{\alpha}^{\pi}\left(\beta\right)\right\rangle .\label{fiberder}
\end{equation}
In canonical coordinates, 
\begin{equation}
\left\langle \mathrm{d}q^{i},\mathbb{F}f\left(p_{k}\,\mathrm{d}q^{k}\right)\right\rangle =\frac{\partial f\circ\left(\varphi^{\ast}\right)^{-1}}{\partial p_{i}}\left(\mathbf{q},\mathbf{p}\right).\label{lce}
\end{equation}
For a quadratic form $\mathfrak{H}:T^{*}Q\rightarrow\mathbb{R}$ {[}see
Eq. \eqref{rmd}{]}, it can be shown that 
\begin{equation}
\mathbb{F}\mathfrak{H}=\rho^{\sharp},\label{fhr}
\end{equation}
and consequently 
\begin{equation}
\mathfrak{H}\left(\alpha\right)=\frac{1}{2}\,\left\langle \alpha,\mathbb{F}\mathfrak{H}\left(\alpha\right)\right\rangle ,\;\;\;\forall\alpha\in T^{*}Q.\label{halfa}
\end{equation}

Denote by $\omega$ the canonical symplectic $2$-form on $T^{*}Q$.
Given two functions $f,g:T^{*}Q\rightarrow\mathbb{R}$, its canonical
Poisson bracket is the function 
\begin{equation}
\left\{ f,g\right\} :=\left\langle \mathrm{d}f,\omega^{\sharp}\left(\mathrm{d}g\right)\right\rangle ,\label{ccpb}
\end{equation}
where $\omega^{\sharp}$ is the inverse of $\omega^{\flat}:TQ\rightarrow T^{*}Q$
and the latter is given by the equation 
\[
\left\langle \omega^{\flat}\left(u\right),v\right\rangle =\omega\left(u,v\right).
\]
In canonical coordinates, omitting $\left(\varphi^{\ast}\right)^{-1}$
for simplicity, 
\begin{equation}
\left\{ f,g\right\} \left(\mathbf{q},\mathbf{p}\right)=\left(\frac{\partial f}{\partial q^{i}}\,\frac{\partial g}{\partial p_{i}}-\frac{\partial g}{\partial q^{i}}\,\frac{\partial f}{\partial p_{i}}\right)\left(\mathbf{q},\mathbf{p}\right).\label{cpb}
\end{equation}

Throughout all of the paper, we will use the following convention
for indices 
\[
\begin{cases}
\text{latin indices } & i,j,k,l=1,\dots,n;\\
\text{latin indices } & a,b,c,d=1,\dots,m;\\
\text{greek indices } & \mu,\nu,\sigma,\rho=1,\dots,n-m.
\end{cases}
\]

\section{Re-writing the matching conditions}

\label{rewr}

Suppose that we have an underactuated Hamiltonian system $(H,\mathcal{W})$
with 
\begin{equation}
H=\mathfrak{H}+h\circ\pi,\label{HHh}
\end{equation}
where $\mathfrak{H}:T^{*}Q\rightarrow\mathbb{R}$ is a quadratic form
{[}see Eq. \eqref{rmd}{]} and $h:Q\rightarrow\mathbb{R}$ is an arbitrary
smooth function. {[}In canonical coordinates, this means that $H$
is given by Eq. \eqref{H}{]}. In other words, we are assuming that
$H$ is \textit{simple}. Note that $\mathbb{F}H=\mathbb{F}\mathfrak{H}=\rho^{\sharp}$
{[}see Eqs. \eqref{fiberder} and \eqref{fhr}{]}. Suppose also\footnote{This will be the general setting and notation from now on.}
that there exists a subbundle $W$ of $T^{*}Q$ of rank $m$ such
that {[}see \eqref{veri}{]} 
\begin{equation}
\mathcal{W}_{\alpha}=\mathsf{vlift}_{\alpha}^{\pi}(W_{\pi\left(\alpha\right)}),\;\;\;\forall\alpha\in T^{*}Q.\label{WvW}
\end{equation}

\begin{remark}
\label{tri}Note that $(H,\mathcal{W})$ can be described by the triple
$(\mathfrak{H},h,W)$. 
\end{remark}
In such a case, according to Ref. \cite{lyap-shap}, the matching
conditions of the Chang's version of the energy shaping method, for
a simple unknown $\hat{H}=\hat{\mathfrak{H}}+\hat{h}\circ\pi$ {[}see
Eq. \eqref{V} for a local expression{]}, are given by 
\begin{equation}
\left\{ \hat{\mathfrak{H}},\mathfrak{H}\right\} \circ\mathbb{F}\hat{\mathfrak{H}}^{-1}\left(\mathfrak{\mathsf{v}}\right)=0,\;\;\;\forall\mathfrak{\mathsf{v}}\in W^{\circ},\label{keg}
\end{equation}
the \textbf{kinetic equation}, and 
\begin{equation}
\left(\mathrm{d}\hat{h}\circ\mathbb{F}\mathfrak{H}-\mathrm{d}h\circ\mathbb{F}\hat{\mathfrak{H}}\right)\circ\mathbb{F}\hat{\mathfrak{H}}^{-1}\left(\mathsf{v}\right)=0,\;\;\;\forall\mathfrak{\mathsf{v}}\in W^{\circ},\label{peg}
\end{equation}
the \textbf{potential equation} {[}see Eqs. $\left(73\right)$ and
$\left(74\right)$ and Remark $18$ of Ref. \cite{lyap-shap}{]}.
\begin{remark}
In order to compare \eqref{keg} and \eqref{peg} with Eqs. $\left(73\right)$
and $\left(74\right)$ of Ref. \cite{lyap-shap}, we must use that
$\mathbb{F}\hat{\mathfrak{H}}\left(W^{\circ}\right)=W^{\bot}$ {[}see
Eqs. \eqref{rort} and \eqref{fhr}{]}.
\end{remark}
Here $\left\{ \cdot,\cdot\right\} $ denotes the canonical Poisson
bracket on $T^{*}Q$ {[}see Eq. \eqref{ccpb}{]}. Equations above
are the intrinsic counterpart of Eqs. \eqref{kec} and \eqref{pec}.
Their data are given by the triple $\left(\mathfrak{H},h,W\right)$,
and their unknown by the pair $(\hat{\mathfrak{H}},\hat{h})$. In
the following, we are going to re-write \eqref{keg} and \eqref{peg}
by redefining the unknown.

\subsection{Intrinsic version}
\begin{lemma}
Given a subbundle $W\subseteq T^{*}Q$, the set of quadratic forms
on $T^{*}Q$ are in bijection with the triples $(\hat{W},\mathfrak{K},\mathfrak{L})$,
where $\hat{W}$ is a complement of $W$, $\mathfrak{K}$ is a quadratic
form on $\hat{W}$, and $\mathfrak{L}$ is a quadratic form on $W$. 
\end{lemma}
\begin{proof}
To any quadratic form $\hat{\mathfrak{H}}:T^{*}Q\rightarrow\mathbb{R}$,
we can assign a triple $(\hat{W},\mathfrak{K},\mathfrak{L})$ with
\begin{equation}
\hat{W}:= W^{\bot},\;\;\;\mathfrak{K}:=\left.\hat{\mathfrak{H}}\right|_{\hat{W}}\;\;\;\textrm{and}\;\;\;\mathfrak{L}:=\left.\hat{\mathfrak{H}}\right|_{W}.\label{udl}
\end{equation}
Here, $W^{\bot}$ means the orthogonal of $W$ w.r.t. $\hat{\mathfrak{H}}$.
Reciprocally, to any triple $(\hat{W},\mathfrak{K},\mathfrak{L})$
as described in the lemma, we can assign the quadratic form 
\begin{equation}
\hat{\mathfrak{H}}:=\mathfrak{K}\circ\hat{\mathfrak{p}}+\mathfrak{L}\circ\mathfrak{p},\label{v}
\end{equation}
where 
\begin{equation}
\hat{\mathfrak{p}}:T^{*}Q\rightarrow\hat{W}\;\;\;\textrm{and}\;\;\;\mathfrak{p}:T^{*}Q\rightarrow W\label{p}
\end{equation}
are the projections related to the decomposition $T^{*}Q=\hat{W}\oplus W$.
It is clear that the map defined by Eq. \eqref{udl} is the inverse
of the map given by \eqref{v}. 
\end{proof}
\begin{remark}
If $\hat{\mathfrak{H}}$ and $(\hat{W},\mathfrak{K},\mathfrak{L})$
are related as in the previous proof, then $\hat{W}$ is always the
orthogonal complement of $W$ with respect to $\hat{\mathfrak{H}}$. 
\end{remark}
\begin{proposition}
Fix a subbundle $W\subseteq T^{*}Q$ and a quadratic form $\mathfrak{H}:T^{*}Q\rightarrow\mathbb{R}$.
If a second quadratic form $\hat{\mathfrak{H}}:T^{*}Q\rightarrow\mathbb{R}$
is a solution of \eqref{keg}, then $\mathfrak{K}:=\left.\hat{\mathfrak{H}}\right|_{\hat{W}}$
is a solution of {[}see Eq. \eqref{p}{]} 
\begin{equation}
\left\{ \mathfrak{K}\circ\hat{\mathfrak{p}},\mathfrak{H}\right\} \left(\mathfrak{\sigma}\right)=0,\;\;\;\forall\sigma\in\hat{W},\label{nkeg}
\end{equation}
where $\hat{W}:= W^{\bot}$ (the orthogonal complement of $W$
w.r.t. $\hat{\mathfrak{H}}$). On the other hand, given a complement
$\hat{W}$ of $W$ and its related projections $\hat{\mathfrak{p}}$
and $\mathfrak{p}$ {[}see Eq. \eqref{p} again{]}, if a quadratic
form $\mathfrak{K}:\hat{W}\rightarrow\mathbb{R}$ satisfies \eqref{nkeg}
then, for every quadratic form $\mathfrak{L}:W\rightarrow\mathbb{R}$,
the function $\hat{\mathfrak{H}}$ given by \eqref{v} satisfies \eqref{keg}. 
\end{proposition}
\begin{proof}
Given both a quadratic form $\hat{\mathfrak{H}}$ and a triple $(\hat{W},\mathfrak{K},\mathfrak{L})$
related as in the previous lemma, let us show that, for all $\sigma\in\hat{W}$,
\begin{equation}
\left\{ \mathfrak{L}\circ\mathfrak{p},\mathfrak{H}\right\} \left(\sigma\right)=0.\label{null}
\end{equation}
To do that, fix a coordinate chart $\left(U,\left(q^{1},\dots,q^{n}\right)\right)$
of $Q$ and a local basis $\left\{ \xi_{1},\dots,\xi_{m}\right\} \subseteq\Omega^{1}\left(U\right)$
of the subbundle $W$. Define 
\begin{equation}
\mathbb{H}^{ij}\left(q\right):=\left\langle \left.\mathrm{d}q^{i}\right|_{q},\mathbb{F}\mathfrak{H}\left(\left.\mathrm{d}q^{j}\right|_{q}\right)\right\rangle ,\quad q\in U,\label{hfh}
\end{equation}
and write 
\[
\mathfrak{p}\left(\left.\mathrm{d}q^{k}\right|_{q}\right)=\mathbb{P}^{ka}\left(q\right)\xi_{a}\left(q\right).
\]
Note that, using \eqref{halfa}, 
\begin{equation}
\mathfrak{H}\left(p_{k}\left.\mathrm{d}q^{k}\right|_{q}\right)=\frac{1}{2}\left\langle p_{k}\left.\mathrm{d}q^{k}\right|_{q},\mathbb{F}\mathfrak{H}\left(p_{l}\left.\mathrm{d}q^{l}\right|_{q}\right)\right\rangle =\frac{1}{2}\,p_{k}p_{l}\,\mathbb{H}^{kl}\left(q\right).\label{hpdq}
\end{equation}
On the other hand, if $\lambda$ is the fibered inner product defining
$\mathfrak{L}$, consider the matrix 
\[
\mathbb{L}_{ab}\left(q\right):=\left\langle \xi_{a}\left(q\right),\lambda^{\sharp}\left(\xi_{b}\left(q\right)\right)\right\rangle ,\;\;\;q\in U.
\]
Then, omitting the dependence on $q$, just for simplicity, we have
that 
\[
\mathfrak{L}\circ\mathfrak{p}\left(p_{k}\mathrm{d}q^{k}\right)=\frac{1}{2}\left\langle \mathfrak{p}\left(p_{k}\mathrm{d}q^{k}\right),\lambda^{\sharp}\left(\mathfrak{p}\left(p_{l}\mathrm{d}q^{l}\right)\right)\right\rangle =\frac{1}{2}\,p_{k}p_{l}\,\mathbb{P}^{ka}\mathbb{P}^{lb}\mathbb{L}_{ab},
\]
and consequently {[}using \eqref{cpb}{]} 
\[
\left\{ \mathfrak{L}\circ\mathfrak{p},\mathfrak{H}\right\} \left(p_{k}\mathrm{d}q^{k}\right)=\left(\frac{\partial}{\partial q^{k}}\left(\mathbb{P}^{ia}\mathbb{P}^{jb}\mathbb{L}_{ab}\right)\mathbb{H}^{kl}-\frac{\partial\mathbb{H}^{ij}}{\partial q^{k}}\mathbb{P}^{ka}\mathbb{P}^{lb}\mathbb{L}_{ab}\right)p_{i}p_{j}p_{l}.
\]
In addition, since $p_{k}\mathrm{d}q^{k}\in\hat{W}$ if and only if
\[
0=\mathfrak{p}\left(p_{k}\mathrm{d}q^{k}\right)=p_{k}\mathbb{P}^{ka}\xi_{a},
\]
which in turn is equivalent to $p_{k}\mathbb{P}^{ka}=0$ for all $a$,
the Eq. \eqref{null} is immediate. To end the proof, it is enough
to note that \eqref{v} and \eqref{null} imply the equality 
\[
\left\{ \hat{\mathfrak{H}},\mathfrak{H}\right\} \left(\sigma\right)=\left\{ \mathfrak{K}\circ\hat{\mathfrak{p}},\mathfrak{H}\right\} \left(\sigma\right),\;\;\;\forall\sigma\in\hat{W},
\]
from which the proposition easily follows. 
\end{proof}
\bigskip{}

So, we can replace the kinetic equation by Eq. \eqref{nkeg}, whose
unknown is a pair $(\mathfrak{K},\hat{W})$: $\hat{W}$ is a complement
of $W$ and $\mathfrak{K}:\hat{W}\rightarrow\mathbb{R}$ is a quadratic
form.

\bigskip{}

Now, let us study the potential equation. Given $\beta\in T^{*}Q$
and $\sigma\in\hat{W}$, on the same fiber, it follows that 
\begin{align*}
\left\langle \beta,\mathbb{F}\left(\mathfrak{L}\circ\mathfrak{p}\right)\left(\sigma\right)\right\rangle  & =\left.\frac{\mathrm{d}}{\mathrm{d}t}\mathfrak{L}\circ\mathfrak{p}(\sigma+t\beta)\right\vert _{t=0}=\left.\frac{\mathrm{d}}{\mathrm{d}t}\mathfrak{L}\circ\mathfrak{p}(t\beta)\right\vert _{t=0}\\
 & =\mathfrak{L}\circ\mathfrak{p}(\beta)\,\left.\frac{\mathrm{d}}{\mathrm{d}t}t^{2}\right\vert _{t=0}=0,
\end{align*}
because $\mathfrak{p}\left(\sigma\right)=0$ and $\mathfrak{L}$ is
a quadratic form. So, $\mathbb{F}\left(\mathfrak{L}\circ\mathfrak{p}\right)\left(\sigma\right)=0$
for all $\sigma\in\hat{W}$, and the potential equation \eqref{peg}
can be written 
\begin{equation}
\left(\mathrm{d}\hat{h}\circ\mathbb{F}\mathfrak{H}-\mathrm{d}h\circ\mathbb{F}\left(\mathfrak{K}\circ\hat{\mathfrak{p}}\right)\right)\left(\sigma\right)=0,\;\;\;\forall\sigma\in\hat{W},\label{epot}
\end{equation}
where only $\hat{W}$ and $\mathfrak{K}$ are involved (the quadratic
form $\mathfrak{L}$ plays no role in either of the two matching conditions).
As a consequence, instead of \eqref{keg} and \eqref{peg}, we can
think of the matching conditions as the Eqs. \eqref{nkeg} and \eqref{epot}
for the unknown $(\mathfrak{K},\hat{h},\hat{W})$, and we shall do
it from now on.

\subsection{Local expressions}

For reasons that will be clear later, we shall concentrate on those
solutions $(\mathfrak{K},\hat{W})$ of the kinetic equations for which
$\hat{W}^{\sharp}=\mathbb{F}\mathfrak{H}(\hat{W})$ is integrable.
The fact that this is alway possible, unless locally, is proved in
the next lemma. 
\begin{lemma}
\label{int}Given a subbundle $W\subseteq T^{*}Q$ of rank $m$ and
a quadratic form $\mathfrak{H}:T^{*}Q\rightarrow\mathbb{R}$, around
every point $q_{0}\in Q$ we can construct, by using only algebraic
manipulations, a local coordinate chart $\left(U,\left(q^{1},\dots,q^{n}\right)\right)$
such that 
\begin{equation}
\hat{W}:=\mathbb{F}\mathfrak{H}^{-1}\left(\mathsf{span}\left\{ \left.\partial\right/\partial q^{1},\dots,\left.\partial\right/\partial q^{n-m}\right\} \right)\label{what}
\end{equation}
is a complement of $W$ along $U$. In particular, $\hat{W}^{\sharp}$
is an integrable subbundle. 
\end{lemma}
\begin{proof}
Given $q_{0}\in Q$, consider a coordinate neighborhood $\left(V,\left(q^{1},\dots,q^{n}\right)\right)$
and a local basis $\left\{ X_{1},\dots,X_{m}\right\} \subseteq\mathcal{\mathfrak{X}}\left(V\right)$
of $W^{\sharp}$ around $q_{0}$. It is clear that 
\[
X_{i}\left(q\right)=\sum_{j=1}^{n}C_{ij}\left(q\right)\left.\frac{\partial}{\partial q^{j}}\right|_{q},\;\;\;q\in V,
\]
being $C_{ij}\left(q\right)$ the coefficients of an $m\times n$
matrix $\mathbb{C}\left(q\right)$ of maximal rank. Reordering the
coordinate functions $q^{i}$'s, if necessary, we can ensure that
the $m\times m$ sub-matrix $\mathbb{S}\left(q_{0}\right)$, given
by the last $m$ columns of $\mathbb{C}\left(q_{0}\right)$, is non-singular.
Then, the vectors 
\[
\left\{ \left.\frac{\partial}{\partial q^{1}}\right|_{q_{0}},\dots,\left.\frac{\partial}{\partial q^{n-m}}\right|_{q_{0}}\right\} 
\]
define a complement of $W_{q_{0}}^{\sharp}$ and, by continuity, the
first $n-m$ coordinate vector fields $\left\{ \left.\partial\right/\partial q^{1},\dots,\left.\partial\right/\partial q^{n-m}\right\} $
span a complement of $W^{\sharp}$ along the open neighborhood $U\subseteq V$
of $q_{0}$ where $\mathbb{S}\left(q\right)$ is non-singular. As
a consequence, the subbundle $\hat{W}$ given by \eqref{what} is
a complement of $W$ along $U$. 
\end{proof}

\begin{remark}
By Frobenius theorem, if $\hat{W}^{\sharp}$ is an integrable subbundle,
then $\hat{W}$ is locally given by \eqref{what} for some coordinate
chart. 
\end{remark}
Now, let us fix a complement $\hat{W}$ of $W$ such that $\hat{W}^{\sharp}=\mathbb{F}\mathfrak{H}(\hat{W})$
is integrable and, given $q_{0}\in Q$, consider a coordinate chart
$\left(U,\varphi=\left(q^{1},\dots,q^{n}\right)\right)$ around $q_{0}$
where $\hat{W}$ is locally given by \eqref{what} (as in the last
lemma). We want to find the form that the matching conditions adopt
in such coordinates. Consider the matrix with entries $\mathbb{H}^{ij}$
given by \eqref{hfh} and define 
\begin{equation}
\mathbb{H}_{ij}:=\left\langle \mathbb{F}\mathfrak{H}^{-1}\left(\frac{\partial}{\partial q^{i}}\right),\frac{\partial}{\partial q^{j}}\right\rangle .\label{h-1}
\end{equation}
We are omitting the dependence on $q$, just for simplicity. Clearly,
\begin{equation}
\mathbb{H}^{ij}\mathbb{H}_{jk}=\delta_{k}^{i}.\label{hh-1}
\end{equation}
Also, the co-vectors 
\[
\sigma_{\mu}:=\mathbb{F}\mathfrak{H}^{-1}\left(\frac{\partial}{\partial q^{\mu}}\right)=\mathbb{H}_{\mu k}\mathrm{d}q^{k}\in\hat{W},\;\;\;\mu=1,\dots,n-m,
\]
give a basis for $\hat{W}$ and we can write 
\begin{equation}
\hat{\mathfrak{p}}\left(\mathrm{d}q^{k}\right)=\hat{\mathbb{P}}^{k\mu}\sigma_{\mu}.\label{pkmu}
\end{equation}
Note that, since $\hat{\mathfrak{p}}\left(\sigma_{\mu}\right)=\sigma_{\mu}$,
\begin{equation}
\mathbb{H}_{\tau k}\hat{\mathbb{P}}^{k\mu}=\delta_{\tau}^{\mu}.\label{phd}
\end{equation}

\begin{remark}
If the (local) forms $\xi_{a}:=\vartheta_{ai}\,\mathrm{d}q^{i}$,
for $a=1,\dots,m$, give a (local) basis for $W$, since $\ker\hat{\mathfrak{p}}=W$,
we also have the identity 
\begin{equation}
\vartheta_{ai}\,\hat{\mathbb{P}}^{i\mu}=0.\label{wp0}
\end{equation}
As a consequence, the matrix $\hat{\mathbb{P}}$ is univocally determined
by the \eqref{phd} and \eqref{wp0}. 
\end{remark}
On the other hand, if $\kappa$ denotes the fibered inner product
on $\hat{W}$ defining $\mathfrak{K}$, consider the matrix 
\[
\mathbb{K}_{\mu\nu}:=\left\langle \sigma_{\mu},\kappa^{\sharp}\left(\sigma_{\nu}\right)\right\rangle .
\]
With this notation, we have that 
\begin{equation}
\mathfrak{K}\circ\hat{\mathfrak{p}}\left(p_{k}\mathrm{d}q^{k}\right)=\frac{1}{2}\left\langle \hat{\mathfrak{p}}\left(p_{k}\mathrm{d}q^{k}\right),\kappa^{\sharp}\left(\hat{\mathfrak{p}}\left(p_{l}\mathrm{d}q^{l}\right)\right)\right\rangle =\frac{1}{2}p_{k}p_{l}\hat{\mathbb{P}}^{k\mu}\hat{\mathbb{P}}^{l\nu}\mathbb{K}_{\mu\nu}.\label{dp}
\end{equation}
Thus, using \eqref{cpb}, \eqref{hpdq} and \eqref{dp}, the Eq. \eqref{nkeg}
translates to 
\[
\left(\frac{\partial}{\partial q^{k}}\left(\hat{\mathbb{P}}^{i\mu}\hat{\mathbb{P}}^{j\nu}\mathbb{K}_{\mu\nu}\right)\mathbb{H}^{kl}-\frac{\partial\mathbb{H}^{ij}}{\partial q^{k}}\hat{\mathbb{P}}^{k\mu}\hat{\mathbb{P}}^{l\nu}\mathbb{K}_{\mu\nu}\right)p_{i}p_{j}p_{l}=0.
\]
In addition, since a generic element of $\hat{W}$ has the form $a^{\tau}\sigma_{\tau}=a^{\tau}\mathbb{H}_{\tau k}\mathrm{d}q^{k}$,
i.e. $p_{k}=a^{\tau}\mathbb{H}_{\tau k}$, using the identities \eqref{hh-1}
and \eqref{phd} we have that the kinetic equation reads 
\begin{equation}
\left(\frac{\partial\mathbb{K}_{\tau_{2}\tau_{3}}}{\partial q^{\tau_{1}}}+\mathbb{G}_{\tau_{1}\tau_{2}\tau_{3}}^{\mu\nu}\mathbb{K}_{\mu\nu}\right)a^{\tau_{1}}a^{\tau_{2}}a^{\tau_{3}}=0,\label{lke}
\end{equation}
with 
\begin{equation}
\mathbb{G}_{\tau_{1}\tau_{2}\tau_{3}}^{\mu\nu}:=\hat{\mathbb{P}}^{k\mu}\delta_{\tau_{1}}^{\nu}\,\frac{\partial\mathbb{H}_{\tau_{2}\tau_{3}}}{\partial q^{k}}+\frac{\partial\left(\hat{\mathbb{P}}^{i\mu}\hat{\mathbb{P}}^{j\nu}\right)}{\partial q^{\tau_{1}}}\,\mathbb{H}_{\tau_{2}i}\mathbb{H}_{\tau_{3}j}.\label{G}
\end{equation}

Now, let us study the potential equation in the above coordinates.
From \eqref{hfh} we know that $\mathbb{F}\mathfrak{H}\left(p_{k}\mathrm{d}q^{k}\right)=p_{k}\mathbb{H}^{kl}\frac{\partial}{\partial q^{l}}$
and, if $p_{k}=a^{\tau}\,\mathbb{H}_{\tau k}$ {[}see \eqref{hh-1}{]},
\begin{equation}
\mathbb{F}\mathfrak{H}\left(a^{\tau}\,\mathbb{H}_{\tau k}\,\mathrm{d}q^{k}\right)=a^{\tau}\frac{\partial}{\partial q^{\tau}}.\label{fha}
\end{equation}
On the other hand, using Eqs. \eqref{lce} and \eqref{dp} we have
that $\mathbb{F}\left(\mathfrak{K}\circ\hat{\mathfrak{p}}\right)\left(p_{k}\,\mathrm{d}q^{k}\right)=p_{k}\,\hat{\mathbb{P}}^{k\mu}\,\hat{\mathbb{P}}^{l\nu}\,\mathbb{K}_{\mu\nu}\,\frac{\partial}{\partial q^{l}}$
and, again, if $p_{k}=a^{\tau}\,\mathbb{H}_{\tau k}$, 
\begin{equation}
\mathbb{F}\left(\mathfrak{K}\circ\hat{\mathfrak{p}}\right)\left(a^{\tau}\,\mathbb{H}_{\tau k}\,\mathrm{d}q^{k}\right)=a^{\mu}\,\hat{\mathbb{P}}^{l\tau}\,\mathbb{K}_{\tau\mu}\,\frac{\partial}{\partial q^{l}}.\label{fdp}
\end{equation}
Thus, from \eqref{fha} and \eqref{fdp}, the potential equation \eqref{epot}
reads 
\[
\left(\frac{\partial\hat{h}}{\partial q^{\mu}}-\frac{\partial h}{\partial q^{k}}\,\hat{\mathbb{P}}^{k\tau}\,\mathbb{K}_{\tau\mu}\right)\,a^{\mu}=0,
\]
which is equivalent to 
\begin{equation}
\frac{\partial\hat{h}}{\partial q^{\mu}}-\frac{\partial h}{\partial q^{k}}\,\hat{\mathbb{P}}^{k\tau}\,\mathbb{K}_{\tau\mu}=0,\;\;\;\mu=1,\dots,n-m.\label{lpe}
\end{equation}
{[}For simplicity, we are identifying $h$ (resp. $\hat{h}$) with
its local representative $h\circ\varphi^{-1}$ (resp. $\hat{h}\circ\varphi^{-1}$){]}.

Summarizing the results of the entire section, we have the next two
theorems. 
\begin{theorem}
\label{1}Consider an underactuated Hamiltonian system satisfying
\eqref{HHh} and \eqref{WvW}, i.e. one defined by a triple $\left(\mathfrak{H},h,W\right)$
(see Remark \ref{tri}). Then, every (simple) solution $\hat{H}=\hat{\mathfrak{H}}+\hat{h}\circ\pi$
of the matching conditions \eqref{keg} and \eqref{peg} is univocally
described by: $\left(\mathbf{i}\right)$ a subbundle $\hat{W}$ complementary
to $W$, $\left(\mathbf{ii}\right)$ a quadratic form $\mathfrak{K}:\hat{W}\rightarrow\mathbb{R}$
solving \eqref{nkeg}, $\left(\mathbf{iii}\right)$ a solution $\hat{h}$
of \eqref{epot}, and $\left(\mathbf{vi}\right)$ a quadratic form
$\mathfrak{L}:W\rightarrow\mathbb{R}$. 
\end{theorem}
\bigskip{}

\begin{theorem}
\label{2} Let $\hat{W}$ be a complement of $W$ such that $\hat{W}^{\sharp}=\mathbb{F}\mathfrak{H}(\hat{W})$
is integrable. Then, in a coordinate chart $\left(U,\left(q^{1},\dots,q^{n}\right)\right)$
of $Q$ satisfying \eqref{what}, the Equations \eqref{nkeg} and
\eqref{epot} translate to \eqref{lke} and \eqref{lpe}, respectively. 
\end{theorem}

\section{Solving the potential equation after solving the kinetic one}

\label{solvp}

In this section, given a solution of the kinetic equation \eqref{nkeg},
we shall study under which conditions a positive solution $\hat{h}$
of the potential equation \eqref{epot} does exist. In fact, we will
develop a sistematic procedure to find (unless locally) an explicit
solution of this equation (up to quadratures). In addition, we will
provide necessary and sufficient conditions to ensure the positivity
of the solution.

\subsection{An integrability condition}

Suppose that the conditions in Theorems \ref{1} and \ref{2} hold,
and that a solution $\mathfrak{K}$ of \eqref{nkeg} is given. We
want to find a (local) solution $\hat{h}$ of the potential equation
\eqref{epot} where $\hat{\mathfrak{p}}$ and $\mathfrak{K}$ are
considered as datum. Given $q_{0}\in Q$, denote by $\left(U,\varphi=\left(q^{1},\dots,q^{n}\right)\right)$
the coordinate chart in which $\hat{W}$ is given by \eqref{what}.
In such coordinates, according to Theorem \ref{2}, Eq. \eqref{epot}
translates to Eq. \eqref{lpe}. As it is well-known, the necessary
and sufficient conditions to integrate these equations are 
\begin{equation}
\frac{\partial}{\partial q^{\mu}}\left(\frac{\partial h}{\partial q^{k}}\,\hat{\mathbb{P}}^{k\tau}\,\mathbb{K}_{\tau\nu}\right)=\frac{\partial}{\partial q^{\nu}}\left(\frac{\partial h}{\partial q^{k}}\,\hat{\mathbb{P}}^{k\tau}\,\mathbb{K}_{\tau\mu}\right),\label{condicionexistencia}
\end{equation}
for all $\mu,\nu\leq n-m$. In global terms they say that 
\begin{equation}
\mathrm{d}\left(\mathrm{d}h\circ\mathbb{F}\left(\mathfrak{K}\circ\hat{\mathfrak{p}}\right)\circ\mathbb{F}\mathfrak{H}^{-1}\right)\left(\mathsf{u},\mathsf{v}\right)=0,\;\;\;\forall\mathsf{u},\mathsf{v}\in\hat{W}^{\sharp}.\label{cf}
\end{equation}
In such a case, the solution not only exists but, furthermore, it
can be computed up to quadratures. We shall give a formula for $\hat{h}$
at the end of this section {[}see \eqref{htech}{]}.

\bigskip{}

\begin{remark}
If $n-m=1$, i.e. if the degree of underactuation is one, the subbundle
$\hat{W}^{\sharp}$ is always integrable because of dimensional reasons.
In addition, the condition \eqref{condicionexistencia} reduces to
a single equation for $\mu=\nu=1$, that immediately holds. Therefore,
if we find a solution $\mathfrak{K}$ of the kinetic equation for
a system with one degree of underactuation, not only there exists
a solution of the single potential equation (as it is already known
in the literature), but even more, we can construct such a solution
(in an appropriate coordinate chart) up to quadratures.
\end{remark}

\subsection{A positivity condition}

Suppose that $q_{0}$ is a critical point of $h$ and, for simplicity,
assume that the above given coordinate chart $\left(U,\varphi\right)$
is centered at $q_{0}$, i.e. $\varphi\left(q_{0}\right)=\left(0,\dots,0\right)=:\mathbf{0}$.
As we have mentioned in the Introduction, one is actually interested
in a solution $\hat{h}$ of \eqref{lpe} which is positive-definite
around $q_{0}$. Identifying $\hat{h}$ with its local representative
$\hat{h}\circ\varphi^{-1}$, this is the same as saying that: 
\begin{enumerate}
\item $\mathbf{0}$ is a critical point of $\hat{h}$ and 
\item the Hessian matrix of $\hat{h}$ at $\mathbf{0}$ is positive-definite. 
\end{enumerate}
If \eqref{condicionexistencia} holds, then a solution $\hat{h}$
of \eqref{lpe} exists and we can use the Method of Characteristics
to construct it. To do that, we must impose a boundary condition,
for instance, along the subset 
\[
S:=\varphi\left(U\right)\cap\left(\left\{ \left(0,\dots,0\right)\right\} \times\mathbb{R}^{m}\right)\subseteq\mathbb{R}^{n}.
\]
So, let us impose on $S$ the condition 
\begin{equation}
\hat{h}(0,\dots,0,s^{1},\dots,s^{m})=\frac{\varpi}{2}\,\sum_{a=1}^{m}\left(s^{a}\right)^{2},\label{vkappa'}
\end{equation}
for some constant $\varpi>0$. It follows from \eqref{lpe}, the boundary
condition above and the fact that $\mathbf{0}$ is a critical point
of $h$ that $\partial\hat{h}/\partial q^{i}(\mathbf{0})=0$ for all
$i=1,\dots,n$. Then, $\mathbf{0}$ is a critical point of $\hat{h}$.
On the other hand, the Hessian of $\hat{h}$ at $\mathbf{0}$, 
\begin{equation}
\text{Hess}(\hat{h})_{ij}(\mathbf{0})=\frac{\partial^{2}\hat{h}}{\partial q^{i}\partial q^{j}}(\mathbf{0}),\label{matrizM}
\end{equation}
can be written 
\[
\text{Hess}(\hat{h})(\mathbf{0})=\begin{bmatrix}\mathbb{M} & \mathbb{A}\\
\mathbb{A}^{t} & \varpi\,\mathbb{I}_{m}
\end{bmatrix},
\]
where: 
\begin{itemize}
\item $\mathbb{I}_{m}$ is the $m\times m$ identity matrix; 
\item $\mathbb{A}$ is the $\left(\left(n-m\right)\times m\right)$-matrix
with entries 
\begin{align}
\mathbb{A}_{\mu a}:= & \left(\text{Hess}(h)_{n-m+a,k}\,\hat{\mathbb{P}}^{k\tau}\,\mathbb{K}_{\tau\mu}\right)(\mathbf{0})\nonumber \\
= & \frac{\partial}{\partial q^{n-m+a}}\left(\frac{\partial h}{\partial q^{k}}\,\hat{\mathbb{P}}^{k\tau}\,\mathbb{K}_{\tau\mu}\right)\left(\mathbf{0}\right),\qquad\mu\leq n-m\text{ and }a\leq m\label{Ai}
\end{align}
and 
\item $\mathbb{M}$ is the square matrix of dimension $n-m$ with 
\begin{equation}
\mathbb{M}_{\mu\nu}:=\left(\text{Hess}(h)_{\mu k}\,\hat{\mathbb{P}}^{k\tau}\,\mathbb{K}_{\tau\nu}\right)(\mathbf{0})=\frac{\partial}{\partial q^{\mu}}\left(\frac{\partial h}{\partial q^{k}}\,\hat{\mathbb{P}}^{k\tau}\,\mathbb{K}_{\tau\nu}\right)\left(\mathbf{0}\right).\label{N}
\end{equation}
\end{itemize}
\begin{remark}
The entries of the matrices $\mathbb{A}$ and $\mathbb{M}$ are obtained
just by differentiating the potential equation \eqref{lpe} and using
Eqs. \eqref{vkappa'} and \eqref{matrizM} and the criticallity of
$\mathbf{0}$ for $h$. 
\end{remark}
\begin{proposition}
Consider a local solution $\hat{h}$ of \eqref{lpe}. If $\text{Hess}(\hat{h})(\mathbf{0})$
is positive-definite, then the matrix $\mathbb{M}$, given by \eqref{N},
is positive-definite. Now, suppose that $\hat{h}$ satisfies \eqref{vkappa'}.
If $\mathbb{M}$ is positive-definite, then there exists a constant
$\varpi$ such that $\text{Hess}(\hat{h})(\mathbf{0})$ is positive-definite
too. 
\end{proposition}
\begin{proof}
The first part of the proposition easily follows from the fact that
$\mathbb{M}$ is an upper-left corner square sub-matrix of $\text{Hess}(\hat{h})(\mathbf{0})$.
So, let us prove the second part. Assume that $\mathbb{M}$ is positive-definite
and take $\mathbf{u}\in\mathbb{R}^{n-m}$ and $\mathbf{w}\in\mathbb{R}^{m}$.
Then 
\begin{align*}
(\mathbf{u},\mathbf{w})\,\begin{bmatrix}\mathbb{M} & \mathbb{A}\\
\mathbb{A}^{t} & \varpi\,\mathbb{I}
\end{bmatrix}\,\begin{bmatrix}\mathbf{u}^{t}\\
\mathbf{w}^{t}
\end{bmatrix} & =\mathbf{u}\,\mathbb{M}\,\mathbf{u}^{t}+2\,\mathbf{u}\,\mathbb{A}\,\mathbf{w}^{t}+\varpi\,\Vert\mathbf{w}\Vert^{2}\\
 & =\left\Vert \mathbf{u}\right\Vert _{\mathbb{M}}^{2}+2\,\mathbf{u}\,\mathbb{A}\,\mathbf{w}^{t}+\varpi\,\Vert\mathbf{w}\Vert^{2}\\
 & \geq\left\Vert \mathbf{u}\right\Vert _{\mathbb{M}}^{2}-2\,\Vert\mathbf{u}\Vert\,\Vert\mathbb{A}\Vert\,\Vert\mathbf{w}\Vert+\varpi\,\Vert\mathbf{w}\Vert^{2},
\end{align*}
where $\Vert\cdot\Vert$ denotes the Euclidean norm in the corresponding
vector space,\footnote{We can also use the operator norm for the matrix $\mathbb{A}$.}
and $\Vert\cdot\Vert_{\mathbb{M}}$ is the norm associated with $\mathbb{M}$.
In particular {[}recall Eq. \eqref{Ai}{]}, 
\begin{equation}
\left\Vert \mathbb{A}\right\Vert ^{2}=\sum_{\mu=1}^{n-m}\sum_{a=1}^{m}\left|\mathbb{A}_{\mu a}\right|^{2}=\sum_{\mu=1}^{n-m}\sum_{a=1}^{m}\left|\frac{\partial}{\partial q^{n-m+a}}\left(\frac{\partial h}{\partial q^{k}}\,\hat{\mathbb{P}}^{k\tau}\,\mathbb{K}_{\tau\mu}\right)\left(\mathbf{0}\right)\right|^{2}.\label{na}
\end{equation}
Since every norm on a finite-dimensional vector space is equivalent
to the Euclidean norm, there exists a positive constant $\alpha$
such that 
\[
\left\Vert \mathbf{u}\right\Vert _{\mathbb{M}}\geq\alpha\left\Vert \mathbf{u}\right\Vert ,\qquad\forall\,\mathbf{u}\in\mathbb{R}^{n-m}.
\]
The constant $\alpha$ may be computed as 
\begin{equation}
\alpha=\min_{\|\mathbf{u}\|=1}\sqrt{\mathbf{u}\,\mathbb{M}\,\mathbf{u}^{t}}=\sqrt{\lambda_{\min}^{\mathbb{M}}},\label{metriccons'}
\end{equation}
where $\lambda_{\min}^{\mathbb{M}}$ is the least eigenvalue of $\mathbb{M}$.
Then, 
\begin{align*}
(\mathbf{u},\mathbf{w})\,\begin{bmatrix}\mathbb{M} & \mathbb{A}\\
\mathbb{A}^{t} & \varpi\,\mathbb{I}
\end{bmatrix}\,\begin{bmatrix}\mathbf{u}^{t}\\
\mathbf{w}^{t}
\end{bmatrix} & \geq\alpha^{2}\left\Vert \mathbf{u}\right\Vert ^{2}-2\,\Vert\mathbf{u}\Vert\,\Vert\mathbb{A}\Vert\,\Vert\mathbf{w}\Vert+\varpi\,\Vert\mathbf{w}\Vert^{2}\\
 & =\alpha^{2}\left(\left\Vert \mathbf{u}\right\Vert ^{2}-2\,\Vert\mathbf{u}\Vert\,\Vert\mathbb{A}\Vert\,\frac{\Vert\mathbf{w}\Vert}{\alpha^{2}}+\varpi\,\frac{\Vert\mathbf{w}\Vert^{2}}{\alpha^{2}}\right).
\end{align*}
If $\mathbb{A}=\mathbf{0}$, this expression is clearly nonnegative
and vanishes only when $\mathbf{u}$ and $\mathbf{w}$ vanish. Suppose
now that $\mathbb{A}\neq\mathbf{0}$. Defining $\beta=\frac{\Vert\mathbb{A}\Vert}{\alpha}$,
we have 
\begin{align*}
(\mathbf{u},\mathbf{w})\,\begin{bmatrix}\mathbb{M} & \mathbb{A}\\
\mathbb{A}^{t} & \varpi\,\mathbb{I}
\end{bmatrix}\,\begin{bmatrix}\mathbf{u}^{t}\\
\mathbf{w}^{t}
\end{bmatrix} & \geq\alpha^{2}\beta^{2}\left(\frac{\Vert\mathbf{u}\Vert^{2}}{\beta^{2}}-2\,\frac{\Vert\mathbf{w}\Vert}{\alpha}\,\frac{\Vert\mathbf{u}\Vert}{\beta}+\frac{\varpi}{\beta^{2}}\,\frac{\Vert\mathbf{w}\Vert^{2}}{\alpha^{2}}\right)\\
 & =\alpha^{2}\beta^{2}\left[\left(\frac{\varpi}{\beta^{2}}-1\right)\left(\frac{\Vert\mathbf{w}\Vert}{\alpha}\right)^{2}+\left(\frac{\Vert\mathbf{w}\Vert}{\alpha}-\frac{\Vert\mathbf{u}\Vert}{\beta}\right)^{2}\right].
\end{align*}
Hence, if we choose $\varpi>\beta^{2}$, i.e. 
\begin{equation}
\varpi>\left(\frac{\Vert\mathbb{A}\Vert}{\alpha}\right)^{2},\label{kappamayor}
\end{equation}
it follows that $\text{Hess}(\hat{h})(\mathbf{0})$ is positive-definite. 
\end{proof}
\bigskip{}

In other words, the previous proposition says that we can find a solution
$\hat{h}$ positive-definite around $\mathbf{0}$ if and only if the
matrix $\mathbb{M}$ is also positive-definite. 
\begin{remark}
Since $q_{0}$ is a critical point of the potential term $h$, we
can find a coordinate-free expression of the matrix $\mathbb{M}$
using the covariant Hessian tensor $\nabla\nabla h$, given by 
\[
\nabla\nabla h(X,Y)=X(Yh)-\left<\mathrm{d}h,(\nabla_{X}Y)\right>,
\]
whose matrix representation at $q_{0}$ is exactly Hessian matrix
of $h$. Indeed, it is easy to see that 
\[
\mathbb{M}_{\mu\nu}=\mathfrak{M}\left(\left.\frac{\partial}{\partial q^{\mu}}\right|_{q_{0}},\left.\frac{\partial}{\partial q^{\nu}}\right|_{q_{0}}\right),
\]
where 
\begin{equation}
\mathfrak{M}\left(\mathsf{u},\mathsf{v}\right)=\nabla\nabla h\left(\mathbb{F}\left(\mathfrak{K}\circ\hat{\mathfrak{p}}\right)\circ\mathbb{F}\mathfrak{H}^{-1}\left(\mathsf{u}\right),\mathsf{v}\right).\label{qfm}
\end{equation}
\end{remark}

\subsection{The integration procedure}

We shall now condense the results of the previous subsections in the
following theorem. Consider again an underactuated system defined
by a triple $\left(\mathfrak{H},h,W\right)$. 

\begin{theorem}
Let $q_{0}$ be a critical point of $h$ and $\left(U,\varphi\right)$
a coordinate neighborhood of $q_{0}$ such that the subbundle $\hat{W}$
given by \eqref{what} is a complement of $W$. Let $\mathfrak{K}$
be a solution of the kinetic equation \eqref{nkeg} for $\hat{W}$.
Then, a solution $\hat{h}$ of the potential equation \eqref{epot}
exists around $q_{0}$ if and only if the condition 
\[
\left.\mathrm{d}\left(\mathrm{d}h\circ\mathbb{F}\left(\mathfrak{K}\circ\hat{\mathfrak{p}}\right)\circ\mathbb{F}\mathfrak{H}^{-1}\right)\right|_{\hat{W}^{\sharp}\times\hat{W}^{\sharp}}=0,
\]
holds {[}see Eq. \eqref{cf}{]}. Moreover, in such a case, $\hat{h}$
can be found up to quadratures. On the other hand, $\hat{h}$ can
be chosen positive-definite around $q_{0}$ if and only if the bilinear
form {[}see Eq. \eqref{qfm}{]} 
\[
\mathfrak{M}=\nabla\nabla h\circ\left.\left(\mathbb{F}\left(\mathfrak{K}\circ\hat{\mathfrak{p}}\right)\circ\mathbb{F}\mathfrak{H}^{-1}\times\text{id}_{TQ}\right)\right|_{\hat{W}^{\sharp}\times\hat{W}^{\sharp}}.
\]
is positive-definite at that point. 
\end{theorem}
Gathering all the results we have presented so far, we can state a
procedure to explicitly construct local solutions of the potential
equation that are positive-definite around $q_{0}$, provided a solution
of the kinetic equation is given. We must proceed as follows: 
\begin{enumerate}
\item find coordinates $(q^{1},\dots,q^{n})$ centered at $q_{0}$ such
that 
\[
\hat{W}:=\mathbb{F}\mathfrak{H}^{-1}\left(\mathsf{span}\left\{ \left.\partial\right/\partial q^{1},\dots,\left.\partial\right/\partial q^{n-m}\right\} \right)
\]
is a complement of $W$ (which can be done just reordering an arbitrary
coordinate system centered at $q_{0}$, as mentioned in the proof
of Lemma \ref{int}); 
\item consider a (local) solution $\mathfrak{K}$ of the kinetic equation
\eqref{nkeg} for $\hat{W}$; 
\item in the coordinates of the step $1$, define the functions $u_{\mu}:=\frac{\partial h}{\partial q^{k}}\,\hat{\mathbb{P}}^{k\tau}\,\mathbb{K}_{\tau\mu}$,
for $\mu=1,\dots,n-m$; 
\item verify that $\frac{\partial u_{\nu}}{\partial q^{\mu}}=\frac{\partial u_{\mu}}{\partial q^{\nu}}$
for all $\mu,\nu\leq n-m$ {[}see Eq. \eqref{condicionexistencia}{]}; 
\item define $\hat{h}$ as 
\begin{equation}
\begin{array}{lll}
\hat{h}(q^{1},\dots,q^{n}) & := & \sum_{\mu=1}^{n-m}\int_{0}^{q^{\mu}}u_{\mu}(0,\dots,0,t,q^{\mu+1},\dots,q^{n})\,\mathrm{d}t\\
\\
 &  & +\frac{\varpi}{2}\,\sum_{a=1}^{m}\left(q^{n-m+a}\right)^{2}
\end{array}\label{htech}
\end{equation}
for some constant $\varpi$; 
\item check that the matrix $\mathbb{M}$ {[}recall Eq. \eqref{N}{]} is
positive-definite, i.e. check that $\lambda_{\min}^{\mathbb{M}}>0$,
\item choose $\varpi$ such that {[}recall Eqs. \eqref{na}, \eqref{metriccons'}
and \eqref{kappamayor}{]} 
\[
\varpi>\frac{\sum_{\mu=1}^{n-m}\sum_{a=1}^{m}\left(\partial u_{\mu}/\partial q^{n-m+a}(\mathbf{0})\right)^{2}}{\lambda_{\min}^{\mathbb{M}}}.
\]
\end{enumerate}
The idea of studying the potential equation assuming we have a solution
of the kinetic equation appeared also in Reference \cite{lewis}.
In such work, the author finds integrability conditions for the potential
matching conditions using Goldschmidt's integrability theory (see
\cite{gold}). Then, assuming such conditions hold and supposing all
the objects involved belong to the category $C^{\omega}$, it is proved
that there exists indeed a solution of these equations. However, the
positivity of such solutions is not analyzed. Here, although our integrability
conditions are similar to those found in \cite{lewis}, our approach
is valid in the category $C^{\infty}$. Moreover, we give necessary
and sufficient conditions in order to ensure that the solution is
positive-definite and, in addition, we show how to build such solution
by computing ordinary integrals in an appropriate coordinate chart.

\section{Solving both the kinetic and the potential equations}

\label{sec:potkincuad}

In this section we are going to apply the steps above to a particular
subclass of underactuated systems.

\subsection{One degree of underactuation}

Assume that $(H,\mathcal{W})$ has one degree of underactuation, i.e.
$\mathcal{W}$ is given by a vector subbundle $W\subset T^{*}Q$ of
rank $m=n-1$. Suppose that the step $1$ was already performed. In
such a case, we can explicitly find a solution $\mathbb{K}$ of the
local kinetic equation \eqref{lke} (corresponding to the step $2$).
Let us see that. To simplify the calculations, we will write 
\[
q^{1}=x,\qquad q^{1+a}=y^{a},\qquad a=1,\dots,n-1,
\]
and 
\[
\mathbb{K}_{11}=K,\qquad\mathbb{G}_{111}^{11}=G,\qquad\text{and}\qquad\hat{\mathbb{P}}^{k1}=\hat{P}^{k},\;\;\;k=1,\dots,n.
\]
Note that, according to \eqref{G}, 
\begin{equation}
G=\frac{\partial\mathbb{H}_{11}}{\partial x}\,\hat{P}^{1}+\frac{\partial\mathbb{H}_{11}}{\partial y^{a}}\,\hat{P}^{1+a}+\frac{\partial\left(\hat{P}^{i}\,\hat{P}^{j}\right)}{\partial x}\,\mathbb{H}_{1i}\,\mathbb{H}_{1j}.\label{Gp}
\end{equation}
Under this notation, Eq. \eqref{lke} reads 
\[
\frac{\partial K}{\partial x}+G\,K=0,
\]
whose general solution is 
\begin{equation}
K\left(x,\mathbf{y}\right)=\xi\left(\mathbf{y}\right)\,e^{-\int_{0}^{x}G(t,\mathbf{y})\,\mathrm{d}t}.\label{K}
\end{equation}
In order for $K$ to define a quadratic form, we must ask $\xi$ to
be a positive function. Following with the step $3$, define 
\begin{equation}
u:= u_{1}=\left(\frac{\partial h}{\partial x}\,\hat{P}^{1}+\frac{\partial h}{\partial y^{a}}\,\hat{P}^{1+a}\right)\,K.\label{u}
\end{equation}
It is clear that the step $4$ is trivial in this case. According
to step $5$, we have 
\begin{equation}
\hat{h}\left(x,\mathbf{y}\right)=\int_{0}^{x}u(t,\mathbf{y})\,\mathrm{d}t+\frac{\varpi}{2}\,\sum_{a=1}^{n-1}\left(y^{a}\right)^{2}.\label{htech2}
\end{equation}
The step $6$ reduces to check that the number 
\[
\mathbb{M}_{11}=\lambda_{\min}^{\mathbb{M}}=\frac{\partial u}{\partial x}(\mathbf{0})
\]
is positive. Finally, step 7 says that we must take 
\begin{equation}
\varpi>\frac{\sum_{a=1}^{n-1}\left(\partial u/\partial y^{a}(\mathbf{0})\right)^{2}}{\partial u/\partial x(\mathbf{0})}.\label{varfi}
\end{equation}

\subsection{The planar inverted double pendulum}

\begin{figure}[h]
\centering \includegraphics[height=0.3\textheight]{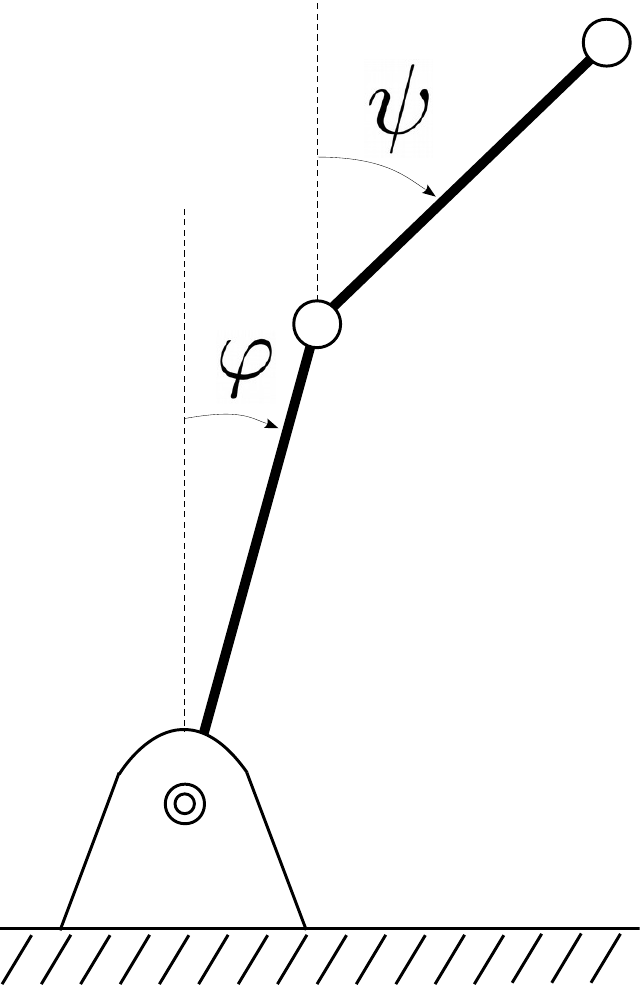}
\caption{Planar inverted double pendulum.}\label{fig:pend} 
\end{figure}

Let us end our work with a concrete example. Consider the manifold
$Q=S^{1}\times S^{1}$ and let $\left(U,\left(\psi,\varphi\right)\right)$
be a system of angular coordinates. Consider on $Q$ the simple Hamiltonian
function $H$ with 
\[
\mathfrak{H}\left(\psi,\varphi,p_{\psi},p_{\varphi}\right)=\frac{1}{2\,m}\,\left(p_{\psi},p_{\varphi}\right)\,\left(\begin{array}{cc}
C & -B\cos(\psi-\varphi)\\
-B\cos(\psi-\varphi) & A
\end{array}\right)\,\left(\begin{array}{c}
p_{\psi}\\
p_{\varphi}
\end{array}\right)
\]
and 
\[
h\left(\psi,\varphi\right)=D_{1}\cos\left(\psi\right)+D_{2}\cos\left(\varphi\right),
\]
where $m:= AC-B^{2}\cos^{2}(\psi-\varphi)$ and $A,B,C,D_1$ and $D_2$ are positive constants. This Hamiltonian corresponds to the system depicted in Figure \ref{fig:pend},
the (planar) \textit{inverted double pendulum}, for appropriate values
of the constants $A,B,C,D_{1}$ and $D_{2}$.\footnote{In fact, if we consider massless bars of lengths $L_{1}$ and $L_{2}$
with particles of masses $m_{1}$ and $m_{2}$ attached to the ends,
the values of these constants are $A=m_{1}L_{1}^{2}+m_{2}L_{2}^{2},B=m_{2}L_{1}L_{2},C=m_{2}L_{2}^{2},D_{1}=m_{2}gL_{2}$
and $D_{2}=(m_{1}+m_{2})gL_{1}$, where $g$ is the acceleration of gravity.} Consider in addition the subbundle $W\subseteq T^{*}Q$ generated
by the $1$-form $\mathrm{d}\varphi$. The latter, together with $H$,
define an underactuated system with one actuator, which produces a
torque around the coordinate $\varphi$. To find a solution of the
matching conditions for $\left(\mathfrak{H},h,W\right)$, let us follow
the steps above. 
\begin{enumerate}
\item Since $W=\left\langle \mathrm{d}\varphi\right\rangle $ along $U$,
then 
\[
W^{\sharp}=\mathbb{F}\mathfrak{H}\left(\left\langle \mathrm{d}\varphi\right\rangle \right)=\left\langle -B\cos(\psi-\varphi)\frac{\partial}{\partial\psi}+A\frac{\partial}{\partial\varphi}\right\rangle 
\]
there. Accordingly, the subbundle $\left\langle \frac{\partial}{\partial\psi}-\gamma\frac{\partial}{\partial\varphi}\right\rangle $
is complementary to $W^{\sharp}$ (shrinking $U$ around $\psi=\varphi=0$,
if needed) if we choose the constant $\gamma\neq\frac{A}{B}$. In what follows we will assume that $\gamma$ is a generic constant which, in the end, we will choose to fulfill our requirements.

Define the coordinates 
\[
x:=\psi\;\;\;\textrm{and}\;\;\;y:=\varphi+\gamma\,\psi,
\]
we have that 
\[
\frac{\partial}{\partial x}=\frac{\partial}{\partial\psi}-\gamma\frac{\partial}{\partial\varphi},\;\;\;\frac{\partial}{\partial y}=\frac{\partial}{\partial\varphi},
\]
and consequently
\begin{align*}
\hat{W} & :=\mathbb{F}\mathfrak{H}^{-1}\left(\left\langle \frac{\partial}{\partial x}\right\rangle \right)
\end{align*}
is complementary to $W$ (along $U$). In this way, the first step
is done. \\ Let us mention that the matrix $\mathbb{H}^{-1}$ representing
$\mathbb{F}\mathfrak{H}^{-1}$ {[}see Eq. \eqref{h-1}{]} in the new
coordinates $\left(x,y\right)$ is
\begin{equation}
\mathbb{H}^{-1}=\left[\begin{array}{cc}
A-2b\gamma+C\gamma^2 & b-\gamma C\\
b-\gamma C & C
\end{array}\right],\label{Hc}
\end{equation}
where 
\begin{equation}
b:= B\cos((1+\gamma)x-y);\label{ab}
\end{equation}
while the matrix $\hat{\mathbb{P}}$ of the projection $\hat{\mathfrak{p}}$,
using $\sigma:=\mathbb{F}\mathfrak{H}^{-1}\left(\frac{\partial}{\partial x}\right)$
as a basis for $\hat{W}$ {[}see Eq. \eqref{pkmu}{]}, is given by
the column vector
\begin{equation}
\hat{\mathbb{P}}=\frac{1}{A-\gamma b}\left(\begin{array}{c}
1\\
\gamma
\end{array}\right).\label{Pc}
\end{equation}
Also, the potential energy $h$ in these coordinates reads
\begin{equation}
h\left(x,y\right)=D_{1}\cos\left(x\right)+D_{2}\,\cos\left(\gamma x-y\right).\label{hch}
\end{equation}
\item The general solution of the kinetic equation is given by \eqref{K},
where, according to \eqref{Gp}, \eqref{Hc}, \eqref{ab} and \eqref{Pc},
\[
G\left(x,y\right)=-\frac{2\gamma^2 b_x}{(1+\gamma)(A-b\gamma)},
\]
where we have considered $\gamma\neq -1$.
Concretely, 
\[
K\left(x,y\right)=\xi\left(y\right)\left(A-b\gamma\right)^{-\frac{2\gamma}{1+\gamma}},
\]
being $\xi$ a positive function. 
\item In this case, the function $u$ defined by \eqref{u} is given by
{[}see \eqref{Pc} and \eqref{hch}{]}
\[
u\left(x,y\right)=-\frac{D_{1}\sin\left(x\right)}{A-b\gamma}K\left(x,y\right).
\]
\item Nothing to do. 
\item According to \eqref{htech2}, 
\begin{align*}
\hat{h}\left(x,y\right)= & -\int_{0}^{x}\frac{D_{1}\sin\left(t\right)}{A-\gamma B\cos((1+\gamma)t-y)}K\left(t,y\right)\,\mathrm{d}t+\frac{\varpi}{2}\,y^{2}.
\end{align*}
\item Since 
\[
M=\frac{\partial u}{\partial x}\left(0,0\right)=\frac{-D_{1}}{A-\gamma B}K\left(0,0\right)
\]
and $K\left(0,0\right)$ and $D_{1}$ are positive, in order for $M$
to be positive, we must choose $\gamma$ such that 
\[
A-\gamma B<0.
\]
This is true if and only if 
\[
\gamma >\frac{A}{B}>0.
\]
In particular, observe that the value $\gamma\neq -1$ we discarded earlier would end up yielding non-positive solutions.

\item Finally, since
\[
\frac{\partial u}{\partial y}\left(0,0\right)=0,
\]
we must take {[}see Eq. \eqref{varfi}{]}
\[
\varpi>0.
\]
\end{enumerate}


\section{Acknowledgments}

S. Grillo and L. Salomone thank CONICET for its financial support.

\end{document}